\documentclass{article}
\usepackage[utf8]{inputenc}
\usepackage{amsmath,amsthm,amssymb,a4wide,hyperref,graphicx}
\usepackage{charter}
\usepackage{color}

\hypersetup{
   colorlinks,%
   citecolor=blue,%
   filecolor=black,%
   linkcolor=red,%
   urlcolor=black
 }

\usepackage[all]{xy}
\usepackage{tikz}
\usetikzlibrary{trees}
\usepgflibrary{arrows}
\usetikzlibrary{positioning,shapes}
\usetikzlibrary{patterns}

\usepackage{times}

\newcommand{\K}{\ensuremath{\textit{K}}}
\newcommand{\Ag}{\ensuremath{\textbf{I}}}
\newcommand{\lr}[1]{\langle #1 \rangle}
\newcommand{\M}{\mathcal{M}}
\newcommand{\knowable}{\Diamond^\K}

\renewcommand{\phi}{\varphi}
\newcommand{\BP}{\textbf{P}}

\theoremstyle{definition}

\newcommand{\weg}[1]{}

\newcommand{\lra}{\leftrightarrow}

\newcommand{\EL}{\textbf{EL}}
\newcommand{\APAL}{\textbf{APAL}}
\newcommand{\PAL}{\textbf{PAL}}
\newcommand{\LUT}{\textbf{LUT}}
\newcommand{\ALUT}{\textbf{ALUT}}

\newcommand{\LUTone}{\textbf{LUT1}}
\newcommand{\LUTtwo}{\textbf{LUT2}}

\newcommand{\LEA}{\textbf{LEA}}
\newcommand{\LAEA}{\textbf{LAEA}}
\newtheorem{theorem}{Theorem}
\newtheorem{lemma}[theorem]{Lemma}
\newtheorem{definition}[theorem]{Definition}

\newtheorem{remark}[theorem]{Remark}
\newtheorem{proposition}[theorem]{Proposition}

\newtheorem{conjecture}[theorem]{conjecture}

\title{Unknown Truths and Unknowable Truths}
\author{Jie Fan\\
\small Institute of Philosophy, Chinese Academy of Sciences;\\
\small School of Humanities, University of Chinese Academy of Sciences  \\
\small \texttt{jiefan@ucas.ac.cn}}
\date{}

\begin{document}

\maketitle

\begin{abstract}
Notions of unknown truths and unknowable truths are important in formal epistemology, which are related to each other in e.g. Fitch's paradox of knowability. Although there have been some logical research on the notion of unknown truths and some philosophical discussion on the two notions, there seems to be no logical research on unknowable truths. In this paper, we propose a logic of unknowable truths, investigate the logical properties of unknown truths and unknowable truths, which includes the similarities of the two notions and the relationship between the two notions, and axiomatize this logic. 
\end{abstract}


\section{Introduction}

This article investigates notions of unknown truths and unknowable truths. A proposition is an {\em unknown truth}, if it is true but unknown; a proposition is an {\em unknowable truth}, if it is true but unknowable. The two notions are related to each other in e.g. Fitch's paradox of knowability, which states that if there is an unknown truth, then there is an unknowable truth.

The notion of unknown truths is important in formal epistemology. For example, it is related to verification thesis, which says that all truths can be known. By the thesis and two uncontroversial principles of knowledge, it follows from the notion that all truths are actually known. In this way, the notion gives rise to a well-known counterexample to the verification thesis. This is the so-called `Fitch's paradox of knowability'. To take another example: the notion gives us an important type of Moore sentences, which is in turn crucial to Moore's paradox, saying that one cannot claim the paradoxical sentence ``$p$ but I do not know it''. A well-known result is that such a Moore sentence is unsuccessful and self-refuting~\cite{hollidayetal:2010,Fan:2019}.

The notion of unknowable truths is crucial in the dispute of realists and anti-realists. Anti-realists holds the aforementioned verificationist thesis, which states that all truths can be known, thereby denying the very possibility of unknowable truths. In contrast, realists believe that there are parts of reality, representable in some conceptually accessible language, that it is impossible for any agent ever to know. As guessed in~\cite[p.~119]{Routley:2010}, ``Perhaps, in truths that cannot be known lies `the mystery at the heart of things'.''

Note that the notion of unknowable truths also gives an important type of Moore sentences, since one cannot claim the paradoxical sentence ``$p$ is true but I cannot know it''. We can also show that such a Moore sentence is unsuccessful and self-refuting.

Some researchers in the literature discuss about whether unknowable truths exist or not. For example, Edgington~\cite{Edgington85} admits that there are
unknowable truths in the sense of Fitch's argument. Horsten~\cite{Horsten:1998} gives an interesting example of unknowable truths, by using a revision of the G\"{o}del sentence.\weg{\footnote{In detail, the revision of the  G\"odel sentence is as follows, which we denote G: 
\begin{center}
G is unknowable.
\end{center}
If G is knowable, then G is true, and thus G is unknowable, contradiction. So G is unknowable. But this is exactly what G says of itself, G must be true.
Thus G is an unknowable truth. 
Note that the argument presupposes the factivity of knowability, which however is disputable. For instance, in Routley~\cite[Fn.~27]{Routley:2010}, it is said that ``under an alethic interpretation of $\Diamond$, it is possible that contingent falsehoods are true and so, in certain cases logically possible that they are known.'' In contrast, Routley thought that the claim holds when $\Diamond$ (that is, `-able') is interpreted as a capacity modality.}}
Cook~\cite{cook2006knights} argues that even if it is impossible for it to be known of any particular sentence
that it is both true and unknowable, we may still know that there are
unknowable truths.\footnote{However, it is shown in~\cite{duke2006unknowable} to be unsuccessful since the argument depends on a paradoxical reasoning.} Besides, it is argued in~\cite[Sec.~2]{Routley:2010} that there are unknowable truths, that is, necessary limits to knowledge. As we will show below, under our interpretation of knowability, there are indeed unknowable truths. Thus our interpretation is in line with the aforementioned philosophical discussion.

In the literature, there have been several interpretations for the notion of knowability. For instance, in~\cite{jfak.fitch:2004,balbianietal:2008}, knowability means `known after an announcement'; in~\cite{Deanetal:2010}, knowability means the existence of proofs; in~\cite{Fara:2010}, knowability means capabilities to know; \weg{in~\cite{Edgington:2010}, it is mentioned that the notion is plausible to be restated as `it is possible that there is evidence' or `it is possible that it is reasonably believed'; }in~\cite{Routley:2010}, it is interpreted as a capacity-knowledge modality `it can be known that';  in~\cite{Wenetal:2011}, knowability means `known after an information update'; whereas in~\cite{Holliday:2018}, knowability means ` dynamically possible knowledge of what was true before the update of the epistemic state of the agent'. In this paper, we follow~\cite{jfak.fitch:2004,balbianietal:2008} to interpret `$\phi$ is knowable' as `$\phi$ is known after some truthful public announcement', so that `$\phi$ is an unknowable truth' is interpreted as `$\phi$ is true and after any truthful public announcement, $\phi$ is unknown'.

Unknowable truths (and unknown truths) is a subjective concept, since it is possible that a proposition is an unknowable truth (an unknown truth, resp.) for an agent but not for another. For instance, consider the true proposition ``it is raining but Ann does not know it''. This proposition is an unknowable truth (an unknown truth, resp.) for Ann, but is {\em not} an unknowable truth (an unknown truth, resp.) for another agent Bob, who may be aware of Ann's ignorance. Therefore we move to the multiagent cases.

Although there have been some logical research on the notion of unknown truths~(see e.g.~\cite{GilbertVenturi:2017,Fan:2022neighborhood}), and also some philosophical discussion on the two notions~\cite{fitch:1963,cook2006knights,duke2006unknowable,Routley:2010}, there seems to be no logical research on unknowable truths. For instance, there have been no axiomatization on the notion of unknowable truths. For another example, it is unnoticed that there are many similarities and relationship between the notions of unknown truths and unknowable truths. In this paper, we propose a logic of unknowable truths,
investigate the logical properties of unknown truths and unknowable truths, which include the similarities and relationship between the two notions, and axiomatize the logic of unknowable truths.

The remainder of the article is organized as follows. Section~\ref{sec.language} introduces the language and semantics of the logic of unknowable truths. Section~\ref{sec.logicalproperties} investigates logical properties of unknown truths and unknowable truths. Section~\ref{sec.soundness} axiomatizes the logic of unknowable truths and shows its soundness. Section~\ref{sec.completeness} demonstrates the completeness. We conclude with some discussion in Section~\ref{sec.conclusion}.

\section{Syntax and Semantics}\label{sec.language}

Throughout the paper, we let $\BP$ be the denumerable set of propositional variables, and let $\Ag$ be a finite set of  agents. 
\begin{definition} Where $p\in\BP$, $i\in\Ag$, the language of the logic of unknowable truths, denoted $\LUT$, is defined recursively as follows:
\[\begin{array}{lll}
\weg{\LEA & \phi::= & p\mid \neg\phi\mid (\phi\land\phi)\mid \bullet_i\phi\\
 \textbf{EL} & \phi::= & p\mid \neg\phi\mid (\phi\land\phi)\mid \K_i\phi\\
\LAEA & \phi::= & p\mid \neg\phi\mid (\phi\land\phi)\mid \bullet_i\phi\mid [\phi]\phi\\
 \textbf{PAL}& \phi::= & p\mid \neg\phi\mid (\phi\land\phi)\mid \K_i\phi\mid [\phi]\phi\\
 \LUTone    &\phi::= & p\mid \neg\phi\mid (\phi\land\phi)\mid \K_i\phi\mid [\phi]\phi\mid U_i^1\phi \\
 \LUTtwo    &\phi::= & p\mid \neg\phi\mid (\phi\land\phi)\mid \bullet_i\phi\mid [\phi]\phi\mid U^2_i\phi \\ }
  \LUT    &\phi::= & p\mid \neg\phi\mid (\phi\land\phi)\mid \K_i\phi\mid [\phi]\phi\mid U_i\phi \\
\end{array}\]
\end{definition}

$\K_i\phi$ is read ``agent $i$ knows that $\phi$'', $[\psi]\phi$ is read ``after every truthfully public announcement of $\psi$, it is the case that $\phi$'', and $U_i\phi$ is read ``$\phi$ is an unknowable truth for $i$''. Other connectives are defined as usual; in particular, $\bullet_i\phi$, read ``$\phi$ is an unknown truth'', abbreviates $\phi\land\neg\K_i\phi$, and $\lr{\psi}\phi$ abbreviates $\neg[\psi]\neg\phi$. 
Without the construct $U_i\phi$, we obtain  public announcement logic $\PAL$; without the construct $[\phi]\phi$ further, we obtain epistemic logic $\EL$.
The language $\LUT$ is interpreted on models.
\begin{definition}
    A {\em model} $\M$ is a triple $\lr{S,\{R_i\mid i\in\Ag\},V}$, where $S$ is a nonempty set of states, each $R_i$ is the accessibility relation for $i$, and $V$ is a valuation. We assume that every $R_i$ is reflexive.\footnote{Note that here for simplicity, we only consider the minimal restriction of knowledge. We can also consider extra properties of knowledge, for instance, transitivity, Euclidicity, etc., in which cases we add the corresponding axioms for knowledge in the axiomatization.} A {\em frame} is a model without valuations. A {\em pointed model} is a pair of a model and a state in the model.
\end{definition}

\begin{definition}
Given a model $\M=\lr{S,\{R_i\mid i\in\Ag\},V}$, the semantics of $\LUT$ is defined recursively as follows.
\[
\begin{array}{|lll|}
\hline
\M,s\vDash p & \iff & s\in V(p)\\
\M,s\vDash \neg \phi & \iff & \M,s\nvDash\phi\\
\M,s\vDash\phi\land\psi &\iff & \M,s\vDash\phi\text{ and }\M,s\vDash\psi\\
\M,s\vDash \K_i\phi & \iff & \text{for all }t,\text{ if }sR_it,\text{ then }\M,t\vDash\phi \\
\M,s\vDash[\psi]\phi & \iff & \text{if } \M,s\vDash\psi\text{ then }\M|_\psi,s\vDash\phi\\
\M,s\vDash U_i\phi&\iff & \M,s\vDash\phi\text{ and for all }\psi\in\EL,\M,s\vDash {[\psi]\neg\K_i}\phi\\
\hline
\end{array}
\]
\end{definition}
where $\M|_\psi$ is the model restriction of $\M$ to $\psi$-states, and the restriction of announced formula in the semantics of $U_i$ is to avoid circularity. We say that $\phi$ is {\em valid over a frame $F$}, if for all models $\M$ based on $F$ and all states $s$, we have $\M,s\vDash\phi$. We say that $\phi$ is {\em valid}, denoted $\vDash\phi$, if it is valid over all (reflexive) frames. We say that $\phi$ is {\em satisfiable}, if there is a pointed model $(\M,s)$ such that $\M,s\vDash\phi$.

Intuitively, $U_i\phi$ means that $\phi$ is true and unknowable for $i$; in more details, $\phi$ is true and after each truthfully public announcement, $\phi$ is unknown to $i$.
\weg{\[
\begin{array}{lll}
\M,s\vDash U^2_i\phi&\iff & \text{for all }\psi\in\LEA,\M,s\vDash {[\psi]\bullet_i}\phi\\
\end{array}
\]}

One may compute the semantics of $\bullet_i$ and $\lr{\psi}\phi$ as follows:
\[
\begin{array}{lll}
\M,s\vDash \bullet_i\phi & \iff & \M,s\vDash\phi\text{ and for some }t, sR_it\text{ and }\M,t\nvDash\phi\\
\M,s\vDash\lr{\psi}\phi & \iff & \M,s\vDash\psi\text{ and }\M|_\psi,s\vDash\phi.\\
\end{array}
\]




\weg{$[\psi]p\lra (\psi\to p)$

$[\psi]\neg\phi\lra (\psi\to\neg[\psi]\phi)$

$[\psi](\phi\land\chi)\lra ([\psi]\phi\land[\psi]\chi)$

$[\psi][\chi]\phi\lra [\psi\land[\psi]\chi]\phi$

${[\psi]\bullet_i}\phi\lra (\psi\to{\bullet_i[\psi]}\phi)$

${[\psi]\K_i}\phi\lra (\psi\to{\K_i[\psi]}\phi)$


$U_i\phi\to {[\psi]\neg\K_i}\phi$, where $\psi\in\EL$

$\dfrac{\eta(\phi\land{[\psi]\neg\K_i}\phi)\text{ for all }\psi\in\mathcal{L}(\bullet)}{\eta(U_i\phi)}$

$\dfrac{\eta({[\psi]\bullet_i}\phi)\text{ for all }\psi\in\LEA}{\eta(U_i\phi)}$ 
}






\weg{\begin{proposition}
For all propositional formulas $\phi$, the semantics of $ U^1_i\phi$ and $U^2_i\phi$ coincide.
\end{proposition}}

\section{Logical properties}\label{sec.logicalproperties}


In this section, we investigate logical properties of unknown truths and unknowable truths. 

\subsection{Similarities}

First, we can summarize some similarities between the logical properties of unknown truths and the logical properties of unknowable truths in the following diagram.
\[ 
\begin{array}{|c||c|c|}
  \hline
  &\text{Unknown Truths}&\text{Unknowable Truths}\\
\hline
  &\bullet_i\phi\to\phi & U_i\phi\to\phi \\
  &\bullet_i (\phi\to\psi)\to (\bullet_i\phi\to \bullet_i\psi) & U_i(\phi\to\psi)\to (U_i\phi\to U_i\psi)\\
 & \bullet_i\phi\land \bullet_i\psi\to \bullet_i(\phi\land\psi)  & U_i\phi\land U_i\psi\to U_i(\phi\land\psi)\\
 &\bullet_i\phi\to \bullet_i{\bullet_i\phi} & U_i\phi\to U_iU_i\phi\\
 \text{validities} &\bullet_i\phi\leftrightarrow \bullet_i{\bullet_i\phi} & U_i\phi\leftrightarrow U_iU_i\phi\\
 & \neg\K_i{\bullet_i\phi} & \neg\K_iU_i\phi \\
 &  \neg {\bullet_i\K_i\phi} ~(\text{transitive})&  \neg U_i\K_i\phi~ (\text{transitive})\\
 & \neg{\bullet_i\neg\K_i\phi}~(\text{Euclidean}) & \neg U_i\neg \K_i\phi ~(\text{Euclidean})\\ 
 & [\bullet_ip]\neg{\bullet_ip} & [U_ip]\neg U_ip\\
  \hline
  \hline
 \text{invalidities} & 
 \neg{\bullet_i\phi}\to \bullet_i\neg{\bullet_i\phi} & \neg U_i\phi\to U_i\neg U_i\phi\\ 
  \hline 
\end{array}
\]

In what follows, we take some of the validities and invalidities as examples, and leave others to the reader.
\weg{\begin{proposition}
$\vDash \bullet_i\phi\to\phi$.
\end{proposition}

\begin{proof}
    Straightforward by definition of $\bullet_i$.
\end{proof}

\begin{proposition}
    $\vDash U_i\phi\to\phi$.
\end{proposition}

\begin{proof}
    Straightforward by definition of $U_i$.
\end{proof}
}

\begin{proposition}
    $\vDash U_i(\phi\to\psi)\to (U_i\phi\to U_i\psi)$.
\end{proposition}

\begin{proof}
    Let $\M=\lr{S,\{R_i\mid i\in\Ag\},V}$ be a model and $s\in S$. Suppose that $\M,s\vDash U_i(\phi\to\psi)$ and $\M,s\vDash U_i\phi$, to show that $\M,s\vDash U_i\psi$. By supposition, we have that $\M,s\vDash \phi$ and for all $\chi\in\EL$, $\M,s\vDash[\chi]\neg \K_i\phi$, and that $\M,s\vDash \phi\to\psi$ and for all $\chi\in\EL$, $\M,s\vDash[\chi]\neg \K_i(\phi\to\psi)$. From $\M,s\vDash\phi$ and $\M,s\vDash\phi\to\psi$, it follows that $\M,s\vDash\psi$. From for all $\chi\in\EL$, $\M,s\vDash[\chi]\neg \K_i(\phi\to\psi)$, we can derive that for all $\chi\in\EL$, $\M,s\vDash [\chi]\neg \K_i\psi$. Therefore, $\M,s\vDash U_i\psi$.
\end{proof}

\weg{\begin{proposition}
    $\vDash \bullet_i (\phi\to\psi)\to (\bullet_i\phi\to \bullet_i\psi)$.
\end{proposition}}

\begin{proposition}
    $\vDash U_i\phi\land U_i\psi\to U_i(\phi\land\psi)$.
\end{proposition}

\begin{proof}
Let $\M=\lr{S,\{R_i\mid i\in\Ag\},V}$ be a model and $s\in S$. Suppose that $\M,s\vDash U_i\phi$ and $\M,s\vDash U_i\psi$, to show that $\M,s\vDash U_i(\phi\land\psi)$. By supposition, $\M,s\vDash\phi$ and for all $\chi\in\EL$, $\M,s\vDash[\chi]\neg\K_i\phi$, and $\M,s\vDash\psi$ and for all $\chi\in\EL$, $\M,s\vDash[\chi]\neg\K_i\psi$. It follows that $\M,s\vDash \phi\land\psi$ and for all $\chi\in\EL$, $\M,s\vDash[\chi]\neg\K_i(\phi\land\psi)$, and therefore $\M,s\vDash U_i(\phi\land\psi)$.
\end{proof}

\weg{\begin{proposition}
    $\vDash \bullet_i\phi\land \bullet_i\psi\to \bullet_i(\phi\land\psi)$.
\end{proposition}}

\weg{\begin{proposition}
    $\vDash U_i(\phi\land\psi)\to U_i\phi\lor U_i\psi$.
\end{proposition}

\begin{proof}
Let $\M=\lr{S,R,V}$ be a model and $s\in S$. Assume that $\M,s\vDash U_i(\phi\land\psi)$, to show that $\M,s\vDash U_i\phi\vee U_i\psi$. By assumption, $\M,s\vDash\phi\land\psi$ and for all $\chi\in\EL$, $\M,s\vDash[\chi]\neg\K_i(\phi\land\psi)$. From the former, it follows that $\M,s\vDash\phi$ and $\M,s\vDash\psi$. From the latter, it follows that for all $\chi\in\EL$, $\M,s\vDash\chi$ implies $\M|_\chi,s\vDash\neg\K_i(\phi\land\psi)$. Note that $\M|_\chi,s\vDash\neg\K_i(\phi\land\psi)$ is equivalent to either $\M|_\chi,s\vDash\neg\K_i\phi$ or $\M|_\chi,s\vDash\neg\K_i\psi$. Then either for all $\chi\in\EL$, $\M,s\vDash\chi$ implies $\M|_\chi,s\vDash\neg\K_i\phi$ or for all $\chi\in\EL$, $\M,s\vDash\chi$ implies $\M|_\chi,s\vDash\neg\K_i\psi$.
\end{proof}}

\weg{The following result is straightforward.
\begin{proposition}
    $\vDash \bullet_i(\phi\land\psi)\to \bullet_i\phi\lor \bullet_i\psi$.
\end{proposition}}

Intuitively, one cannot know the unknowable truths. This is because unknowable truths are themselves unknowable truths, as shown below. One can show that it is equivalent to the statement that it is impossible for it to be known of any particular sentence
that it is both true and unknowable, where the latter is argued by Cook~\cite{cook2006knights}.

\begin{proposition}
    $\vDash U_i\phi\to U_iU_i\phi$, and thus $\vDash U_i\phi\leftrightarrow U_iU_i\phi$.
\end{proposition}

\begin{proof}
   Let $\M=\lr{S,\{R_i\mid i\in\Ag\},V}$ be a model and $s\in S$. Suppose that $\M,s\vDash U_i\phi$, to show that $\M,s\vDash U_iU_i\phi$. By supposition, it suffices to show that for all $\psi\in \EL$, $\M,s\vDash[\psi]\neg \K_iU_i\phi$.

   If not, then for some $\psi\in\EL$ we have $\M,s\vDash\lr{\psi} \K_iU_i\phi$, thus $\M|_\psi,s\vDash \K_iU_i\phi$. From this it follows that for all $t$ such that $sR_it$, $\M|_\psi,t\vDash U_i\phi$, and then $\M|_\psi,t\vDash\phi$, and for all $\chi\in\EL$, $\M|_\psi,t\vDash[\chi]\neg \K_i\phi$, which implies that $\M|_\psi,t\vDash\neg \K_i\phi$. Since $sR_is$, we obtain $\M|_\psi,s\vDash \neg \K_i\phi$. On the other hand, since for all $t$ such that $sR_it$ we have $\M|_\psi,t\vDash\phi$, we infer that $\M|_\psi,s\vDash\K_i\phi$: a contradiction.

   Based on the above analysis, we conclude that $\M,s\vDash U_iU_i\phi$, as desired.
\end{proof}

\weg{If it is an unknown truth that $p$, it is an unknown truth that it is an unknown truth that $p$.

\begin{proposition}
    $\vDash \bullet_i\phi\to \bullet_i{\bullet_i\phi}$, and thus $\vDash \bullet_i\phi\leftrightarrow \bullet_i{\bullet_i\phi}$.
\end{proposition}}

Intuitively, one cannot know the unknowable truths, since otherwise one would know the truths that cannot be known, which is impossible. In other words, unknowable truths are necessary limits to knowledge. This follows immediately from the result below.
\begin{proposition}\label{prop.notknowunknowable}
    $\vDash \neg\K_iU_i\phi$.
\end{proposition}

\begin{proof}
Suppose, for a contradiction, that there is a model $\M=\lr{S,\{R_i\mid i\in\Ag\},V}$ and $s\in S$ such that $\M,s\vDash\K_iU_i\phi$. Then for all $t$, if $sR_it$ then $\M,t\vDash U_i\phi$, thus $\M,t\vDash\phi$ and for all $\psi\in\EL$, $\M,t\vDash[\psi]\neg\K_i\phi$. On one hand, we can obtain that $\M,s\vDash\K_i\phi$; on the other hand, by letting $t$ be $s$ and $\psi$ be $\top$, we infer that $\M,s\vDash\neg\K_i\phi$: a contradiction.
\end{proof}

\weg{Similarly, one does not know an unknown truth.
\begin{proposition}
$\vDash \neg\K_i{\bullet_i\phi}$.
\end{proposition}}

\begin{remark}
If we assume the property of transitivity, then $\neg U_i\K_i\phi$ is valid, which says that all knowledge is not an unknowable truth. The proof goes as follows.
Assume, for a contradiction, that there is a model $\M=\lr{S,\{R_i\mid i\in\Ag\},V}$ and a state $s\in S$ such that $\M,s\vDash U_i\K_i\phi$. Then $\M,s\vDash\K_i\phi$, which implies that $\M,s\vDash\K_i\K_i\phi$ due to transitivity. We have also that for all $\psi\in\EL$, $\M,s\vDash[\psi]\neg\K_i\K_i\phi$. By letting $\psi$ be $\top$, we infer that $\M,s\vDash\neg\K_i\K_i\phi$. Contradiction. Similarly, if we assume the property of Euclidicity, we can show that $\neg U_i\neg \K_i\phi$ is valid, which says that all non-knowledge is not an unknowable truth. Similar results apply to the case for the notion of unknown truths: $\neg{\bullet_i\K_i\phi}$ is valid over transitive frames, and $\neg{\bullet_i\neg\K_i\phi}$ is valid over Euclidean frames.
\end{remark}

The following states that {\em not} all non-unknowable-truths are themselves unknowable truths.
\begin{proposition}
    $\nvDash \neg U_i\phi\to U_i\neg U_i\phi$.
\end{proposition}

\begin{proof}
    We show that $\nvDash \neg U_i\top \to U_i\neg U_i\top$. First, it should be clear that $\vDash \neg U_i\top$, since there is an announcement (say $\top$) such that after that the agent knows $\top$. However, $\nvDash U_i\neg U_i\top$, since validities are always knowable truths, thus not unknowable truths (we will defer the statement to Prop.~\ref{prop.validityisknowable}).
\end{proof}

\weg{Similarly, we can show that
\begin{proposition}
$\nvDash \neg{\bullet_i\phi}\to \bullet_i\neg{\bullet_i\phi}$.
\end{proposition}}

As mentioned in the introduction, Moore sentences such as $\bullet_ip$ and $U_ip$ are unsuccessful and self-refuting. It is shown in~\cite{Fan:2019} that $\bullet_ip$ is unsuccessful and self-refuting. Here we show that $U_ip$ is also unsuccessful and self-refuting.

\begin{proposition}
$\vDash[U_ip]
\neg U_ip$, and thus $\nvDash [U_ip]U_ip$.
\end{proposition}

\begin{proof}
    Suppose that $\M,s\vDash U_ip$, to show that $\M|_{U_ip},s\nvDash U_ip$. By supposition, $\M,s\vDash p$. Thus $\M|_{U_ip}\subseteq \M|_p$, and then $\M|_{U_ip}|_{U_ip}\subseteq \M|_p|_{U_ip}$. Since $\M|_p,s\vDash \K_ip$, we obtain that $\M|_p,s\nvDash U_ip$, that is, $s\notin \M|_p|_{U_ip}$. It then follows that $s\notin \M|_{U_ip}|_{U_ip}$, and therefore $\M|_{U_ip},s\nvDash U_ip$.
\end{proof}

One may ask if there are any differences between the logical properties of unknown truths and the logical properties of unknowable truths. The answer is positive. We have seen from~\cite{Fan:2019} that $\bullet_i(\phi\land\psi)\to \bullet_i\phi\lor \bullet_i\psi$ is valid. In contrast, $U_i(\phi\land\psi)\to U_i\phi\vee U_i\psi$ is invalid.
\begin{proposition}
    $\nvDash U_i(\phi\land\psi)\to U_i\phi\lor U_i\psi$.
\end{proposition}

\begin{proof}
Consider the following model $\M$:
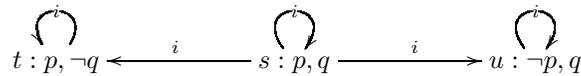
\[
\xymatrix{t: p,\neg q\ar@(ur,ul)|{i} && s:p,q\ar@(ul,ur)|{i}\ar[ll]_i\ar[rr]^i && u:\neg p,q\ar@(ul,ur)|{i}}
\]
We show that $\M,s\vDash U_i(\neg \K_ip\land \neg \K_iq)$ but $\M,s\nvDash U_i\neg \K_ip$ and $\M,s\nvDash U_i\neg\K_iq$.
\begin{itemize}
\item $\M,s\vDash U_i(\neg \K_ip\land \neg \K_iq)$. Clearly, $\M,s\vDash\neg\K_ip\land\neg\K_iq$. Suppose, for reductio, that there exists $\psi\in\EL$ such that  $\M,s\vDash\lr{\psi}\K_i(\neg\K_ip\land\neg\K_iq)$, and then $\M,s\vDash\psi$ and $\M|_\psi,s\vDash\K_i(\neg\K_ip\land\neg\K_iq)$. Because the announcement is interpreted via world-elimination, we consider the following three cases.
\begin{itemize}
    \item $t$ is retained in the updated model $\M|_\psi$. On the one hand, since $\M|_\psi,s\vDash\K_i(\neg\K_ip\land\neg\K_iq)$ and $sR_it$, we obtain $\M|_\psi,t\vDash\neg\K_ip$; on the other hand, as $t$ can only see itself and $t\vDash p$, we have $\M|_\psi,t\vDash\K_ip$. A contradiction.
    \item $u$ is retained in the updated model $\M|_\psi$. Similar to the first case, we can show that $\M|_\psi,u\vDash\K_iq$ and $\M|_\psi,u\vDash\neg\K_iq$, a contradiction.
    \item Neither $t$ nor $u$ is retained in $\M|_\psi$. Then in $\M|_\psi$, there is only a state $s$ which is reflexive. Similarly to the first case, we can also show that $\M|_\psi,s\vDash\K_ip$ and $\M|_\psi,s\vDash\neg\K_ip$, a contradiction.
\end{itemize}
\item $\M,s\nvDash U_i\neg\K_ip$. We have seen that $\M,s\vDash\neg\K_ip$. It suffices to show that for some $\psi\in\EL$ we have $\M,s\vDash\lr{\psi}\K_i\neg\K_ip$. This is indeed the case, because announcing $q$ makes $t$ be deleted, and in the remaining model, $\neg\K_ip$ is true at $s$ and $u$.
\item $\M,s\nvDash U_i\neg\K_iq$. The proof is analogous to that of $\M,s\nvDash U_i\neg\K_ip$, by announcing $p$ instead.
\end{itemize}
\end{proof}

\subsection{Interactions}

In what follows, we study the interactions between unknown truths and unknowable truths. First, every unknowable truth is an unknown truth, but not vise versa. This indicates that the notion of unknowable truths is stronger than that of unknown truths. 
\begin{proposition}\label{prop.unknowabletounknown}
$\vDash U_i\phi\to\bullet_i\phi$, and $\nvDash \bullet_i\phi\to U_i\phi$.
\end{proposition}

\begin{proof}
Let $\M=\lr{S,R,V}$ be a model and $s\in S$. Assume that $\M,s\vDash U_i\phi$, then $\M,s\vDash\phi$ and for all $\psi\in\EL$, $\M,s\vDash[\psi]\neg\K_i\phi$. Letting $\psi=\top$, we obtain that $\M,s\vDash\neg\K_i\phi$, and therefore $\M,s\vDash\bullet_i\phi$.

For the invalidity, consider an instance where $\phi=p$. It is possible that $p$ is true but possibly false, but $p$ must be knowable. For the strict proof, 
consider the following model:
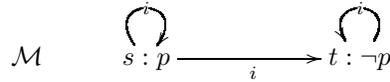
\[
\xymatrix{\M& s:p\ar@(ul,ur)|i\ar[rr]_i && t:\neg p\ar@(ur,ul)|i}
\]
We will show that $\M,s\nvDash {\bullet_ip}\to U_ip$.

First, since $\M,s\vDash p$ and $sR_it$ and $\M,t\nvDash p$, we have $\M,s\vDash \bullet_ip$. Moreover, because $\M,s\vDash \lr{p}\K_ip$ we have that $\M,s\nvDash U_ip$.
\end{proof}

\weg{\begin{proposition}
$\bullet_i\phi\to {U_i\bullet_i}\phi$ is valid. (Wrong!)
\end{proposition}}

Also, intuitively, one cannot know the unknown truths. This is because unknown truths are themselves unknowable truths. In other words, if it is an unknown truth that $p$, it is an unknowable truth that it is an unknown truth that $p$, which 
is argued in~\cite[Thm.~2]{fitch:1963} (see also~\cite[p.~154]{Williamson:1987}).
\begin{proposition}\label{prop.unknowntruthsareunknowable}
     $\vDash \bullet_i\phi\to {U_i\bullet_i}\phi$, thus $\vDash \bullet_i\phi\leftrightarrow {U_i\bullet_i}\phi$.
\end{proposition}

\begin{proof}
Let $\M=\lr{S,\{R_i\mid i\in\Ag\},V}$ be a model and $s\in S$. Suppose that $\M,s\vDash\bullet_i\phi$, to show that $\M,s\vDash U_i{\bullet_i\phi}$. For this, it suffices to prove that for all $\psi\in\EL$, $\M,s\vDash[\psi]\neg\K_i{\bullet_i\phi}$.

If not, then for some $\psi\in \EL$, we have $\M,s\vDash\lr{\psi}\K_i{\bullet_i\phi}$, then $\M|_\psi,s\vDash\K_i{\bullet_i\phi}$. Then for all $t$ such that $sR_it$, we can obtain that $\M|_\psi,t\vDash\phi$ and $\M|_\psi,t\nvDash\K_i\phi$. Then $\M|_\psi,s\vDash\K_i\phi$ and $\M|_\psi,s\nvDash\K_i\phi$ (since $sR_is$),  a contradiction.
\end{proof}


From Prop.~\ref{prop.unknowabletounknown} and Prop.~\ref{prop.unknowntruthsareunknowable}, we can see that $\vDash U_i\phi\to U_i{\bullet_i\phi}$. Despite this, the converse formula is not valid. This follows directly from Prop.~\ref{prop.unknowabletounknown} and Prop.~\ref{prop.unknowntruthsareunknowable}.
\begin{proposition}\label{prop.notunknowablebullettounknowable}
$\nvDash U_i{\bullet_i\phi}\to U_i\phi$.
\end{proposition}

\weg{\begin{proof}
By Prop.~\ref{prop.unknowntruthsareunknowable}, it suffices to show that $\nvDash \bullet_i\phi\to U_i\phi$. Consider the following model:
\[
\xymatrix{\M& s:p\ar@(ul,ur)|i\ar[rr]_i && t:\neg p\ar@(ur,ul)|i}
\]
We will show that $\M,s\nvDash U_i{\bullet_ip}\to U_ip$.

First, since $\M,s\vDash p$ and $sR_it$ and $\M,t\nvDash p$, we have $\M,s\vDash \bullet_ip$. Moreover, because $\M,s\vDash \lr{p}\K_ip$ we have that $\M,s\nvDash U_ip$.
\end{proof}}

The following states that unknowable truths are themselves unknown truths. In other words, if it is an unknowable that $\phi$, then it is an unknown truth that it is an unknowable truth that $\phi$.
\begin{proposition}
$\vDash U_i\phi\to {\bullet_iU_i}\phi$, thus $\vDash U_i\phi\leftrightarrow {\bullet_iU_i}\phi$.    
\end{proposition}

\begin{proof}
Let $\M=\lr{S,\{R_i\mid i\in\Ag\},V}$ be a model and $s\in S$. Suppose that $\M,s\vDash U_i\phi$, to demonstrate that $\M,s\vDash \bullet_i{U_i\phi}$. For this, it suffices to show that for some $t$, $sR_it$ and $\M,t\nvDash U_i\phi$.

If not, then for all $t$ such that $sR_it$, we have $\M,t\vDash U_i\phi$, thus $\M,t\vDash\phi$ and for all $\psi\in\EL$, $\M,t\vDash[\psi]\neg\K_i\phi$, which gives us $\M,t\vDash\neg\K_i\phi$. Since $sR_is$, we derive that $\M,s\vDash\neg\K_i\phi$; however, from $\M,t\vDash\phi$ for all $t$ with $sR_it$, it follows that $\M,s\vDash\K_i\phi$: a contradiction.
\end{proof}

Fitch's paradox of knowability states that if all truths are knowable, then all truths are actually known. This can be shown as follows.
\begin{proposition}
    $\vDash \neg U_i\phi$ for all $\phi$, then $\vDash\neg {\bullet_i\phi}$ for all $\phi$.
\end{proposition}

\begin{proof}
    It suffices to show that $\nvDash \neg U_i\phi$ for some $\phi$, which intuitively says that there are unknowable truths. This follows from Prop.~\ref{prop.unknowntruthsareunknowable}, since $\bullet_i\phi$ is satisfiable. 
\end{proof}

Although not every truth (e.g. unknown truths) is knowable, every {\em logical truth} (namely, validity) is knowable.

\begin{proposition}\label{prop.validityisknowable}
    If $\vDash \phi$, then $\vDash\neg U_i\phi$.
\end{proposition}

\section{Axiomatization and Soundness}\label{sec.soundness}

In this section, we will give a proof system for $\LUT$ and show its soundness. For this, we need a notion of {\em admissible form}. This notion is originated from~\cite[pp.~55–56]{goldblatt:1982}, which is called ‘necessity forms’ in~\cite{balbianietal:2007,balbianietal:2008,philippe.corrected:2015,balbianietal:2015}.

\begin{definition}[Admissible forms] Let $\phi\in\textbf{LUT}$ and $i\in\Ag$. Admissible forms $\eta(\sharp)$ are defined recursively in the following.
\[
\begin{array}{ll}
    \eta(\sharp)::= & \sharp\mid \phi\to\eta(\sharp)\mid {\K_i}\eta(\sharp)\mid [\phi]\eta(\sharp) \\
\end{array}
\]
\end{definition}

Note that $\sharp$ is a placeholder, not a formula. By replacing $\sharp$ in an admissible form $\eta(\sharp)$ with a formula $\chi$, we obtain a formula $\eta(\psi)$, defined as follows.
\[
\begin{array}{lll}
    \sharp(\chi) & = & \chi \\
    (\phi\to \eta(\sharp))(\chi) & = & \phi\to\eta(\chi)\\    {\K_i}\eta(\sharp)(\chi) & = & \K_i\eta(\chi)\\
    ([\phi]\eta(\sharp))(\chi) & = & [\phi]\eta(\chi)\\
\end{array}
\]

With the previous preparation, we now axiomatize the logic of unknowable truths $\LUT$.

\begin{definition}
The proof system $\mathbb{LUT}$ consists of the following axioms and inference rules.
\[
\begin{array}{ll}
   \text{PL} & \text{All instances of tautologies}\\
   \text{K} & \K_i(\phi\to\psi)\to (\K_i\phi\to\K_i\psi)\\
   \text{KA} & [\chi](\phi\to\psi)\to ([\chi]\phi\to [\chi]\psi)\\
   \text{T} & \K_i\phi\to \phi\\
   \text{AP}  & [\psi]p\lra (\psi\to p) \\
    \text{AN} & [\psi]\neg\phi\lra (\psi\to\neg[\psi]\phi)\\
    \text{AC} & [\psi](\phi\land\chi)\lra ([\psi]\phi\land[\psi]\chi)\\
    \text{AK} & [\psi]\K_i\phi\lra (\psi\to\K_i[\psi]\phi)\\
    \text{AA} & [\psi][\chi]\phi\lra [\psi\land[\psi]\chi]\phi\\
   \text{AU} & U_i\phi\to \phi\land [\psi]\neg\K_i\phi, \text{ where } \psi\in\EL\\
   \text{MP} & \dfrac{\phi~~~\phi\to\psi}{\psi}\\
   \text{GEN} & \dfrac{\phi}{\K_i\phi}\\
   \text{GENA} & \dfrac{\phi}{[\chi]\phi}\\
   \text{RU} & \dfrac{\eta(\phi\land[\psi]\neg \K_i\phi)\text{ for all }\psi\in\EL}{\eta(U_i\phi)}\\
\end{array}
\]
\end{definition}

The novel thing is the axiom AU and inference rule RU, which together characterize the semantics of $U_i$ in a certain way: AU says that if $\phi$ is an unknowable truth for $i$, then $\phi$ is true and after every truthfully public announcement, $\phi$ is unknown to $i$; RU roughly says that if $\phi$ is true but unknown to $i$ after every truthfully public announcement, then $\phi$ is an unknowable truth for $i$.
Other axioms and inference rules are familiar from public announcement logic, cf.~e.g.~\cite{hvdetal.del:2007}. 

We say a formula $\phi$ is a {\em theorem} of $\mathbb{LUT}$, or $\phi$ is {\em provable} in $\mathbb{LUT}$, notation: $\vdash\phi$, if $\phi$ is either an instantiation of some axiom, or obtained from axioms with an application of some inference rule. Let $\textbf{Thm}$ be the set of all theorems of $\mathbb{LUT}$.

\begin{proposition}
$\mathbb{LUT}$ is sound with respect to the class of all frames.
\end{proposition}

\begin{proof}
We only consider the axiom AU and the inference rule RU. The validity of axiom AC is straightforward by semantics. For the soundness of RU, we need to show that if $\vDash \eta(\phi\land[\psi]\neg \K_i\phi)\text{ for all }\psi\in\EL$, then $\vDash \eta(U_i\phi)$. To do this, we show the following stronger result:
\begin{center}
$(\star)$~~~~~~for all $(\M,s)$, if $\M,s\vDash\eta(\phi\land[\psi]\neg \K_i\phi)\text{ for all }\psi\in\EL$, then $\M,s\vDash \eta(U_i\phi)$.
\end{center}
The proof proceeds by induction on the structure of admissible forms. The base case and the inductive cases $\phi\to\eta(\sharp)$ and $\K_i\eta(\sharp)$ and $[\phi]\eta(\sharp)$ follow from the semantics and the induction hypothesis.
\end{proof}

In the next section, we will show the completeness of $\mathbb{LUT}$, where the canonical model will be based on a notion of {\em maximal consistent theories}, rather than the more familiar notion of maximal consistent sets. This is because we need to use the condition of closure under RU, which is indispensible in the completeness proof.
\begin{definition}[theory, consistent theory, maximal theory, maximal consistent theory]
Given a set of formulas $s$, we say it is a {\em theory}, if the following holds:

(a) $\textbf{Thm}\subseteq s$, and 

(b) $s$ is closed under $\text{MP}$ and $\text{RU}$.\\
Moreover, if $\bot\notin s$, then $s$ is a {\em consistent}; if for all $\phi$, either $\phi\in s$ or $\neg\phi\in s$, then $s$ is a {\em maximal}; if $s$ is consistent and maximal, then $s$ is {\em maximally consistent}.
\end{definition}

One may check that $\textbf{Thm}$ is the smallest theory.

Define $s+\phi=\{\psi\mid \phi\to\psi\in s\}$.

\begin{proposition}\label{prop.theory-property}
Let $s$ be a theory and $\phi$ be a formula. Then 

(1) $s+\phi$ is a theory that contains $s\cup\{\phi\}$, and 

(2) $s+\phi$ is consistent iff $\neg\phi\notin s$.
\end{proposition}

\begin{proof}
(1) Suppose that $s$ is a theory. First, we show that $s+\phi$ is a theory. 
\begin{itemize}
    \item Contains $\textbf{Thm}$: given any $\psi\in \textbf{Thm}$, we have $\phi\to\psi\in \textbf{Thm}$. Since $\textbf{Thm}\subseteq s$ (as $s$ is a theory), $\phi\to\psi\in s$, and thus $\psi\in s+\phi$.
    \item Closed under $\text{MP}$: assume that $\psi,\psi\to\chi\in s+\phi$, then $\phi\to \psi\in s$ and $\phi\to(\psi\to\chi)\in s$. Since $s$ contains $\textbf{Thm}$ and is closed under $\text{MP}$, we infer that $\phi\to\chi\in s$, and therefore $\chi\in s+\phi$.
    \item Closed under $\text{RU}$: suppose that $\eta(\chi\land[\psi]\neg\K_i\chi)\in s+\phi$ for all $\psi\in\EL$, then $\phi\to \eta(\chi\land[\psi]\neg\K_i\chi)\in s$. Because $s$ is closed under $\text{RU}$ for the admissible form $\phi\to\eta(\sharp)$, we derive that $\phi\to\eta(U_i\chi)\in s$, so $\eta(U_i\chi)\in s+\phi$.
\end{itemize}

Next, we show that $s\cup\{\phi\}\subseteq s+\phi$. Given any $\psi\in s\cup\{\phi\}$, either $\psi\in s$ or $\psi=\phi$. If $\psi\in s$, then as $\psi\to(\phi\to\psi)\in \textbf{Thm}$ and $\textbf{Thm}\subseteq s$, $\psi\to(\phi\to\psi)\in s$; since $s$ is closed under $\text{MP}$, $\phi\to\psi\in s$, that is, $\psi\in s+\phi$. If $\psi=\phi$, then since $\phi\to\phi\in s$, we immediately have $\phi\in s+\phi$, that is, $\psi\in s+\phi$. Therefore, $s\cup\{\phi\}\subseteq s+\phi$.

(2) Suppose that $\neg\phi\in s$, then by (1), $\neg\phi\in s+\phi$. Moreover, by (1) again, $\phi\in s+\phi$. We have also $\neg\phi\to (\phi\to \bot)\in s+\phi$. Since $s+\phi$ is closed under $\text{MP}$ (by (1) again), we conclude that $\bot\in s+\phi$, and therefore $s+\phi$ is inconsistent.

Conversely, assume that $s+\phi$ is inconsistent, then $\bot\in s+\phi$, and thus $\phi\to\bot\in s$. Since $(\phi\to\bot)\to \neg\phi\in s$ and $s$ is closed under $\text{MP}$, it follows that $\neg\phi\in s$.
\end{proof}

Lindenbaum’s Lemma can be proved as~\cite[Lem.4.12]{balbianietal:2008}.
\begin{lemma}[Lindenbaum's Lemma]\label{lem.lindenbaumlemma}
Every consistent theory can be extended to a maximal consistent theory.
\end{lemma}

\section{Completeness}\label{sec.completeness}

To demonstrate the completeness of $\mathbb{LUT}$, we adapt the method in~\cite{balbianietal:2015} for APAL. The point is as follows. In the proof of the truth lemma, the hardest part is the cases involving $U_i\phi$. To handle this, we need to `reduce' the case $U_i\phi$ to the easy-to-fix case $\phi\land[\psi]\neg\K_i\phi$ (where $\psi\in\EL$), via a formula-based complexity measure $<^S_d$ (defined below), and then use the $<^S_d$-induction hypothesis to deal with the cases involving $U_i\phi$ in the truth lemma. 

\begin{definition}[Canonical Model]
The {\em canonical model} for $\mathbb{LUT}$ is a tuple $\M^c=\lr{S^c,\{R^c_i\mid i\in\Ag\},V^c}$ where
\begin{itemize}
    \item $S^c=\{s\mid s\text{ is a maximal consistent theory for }\mathbb{LUT}\}$.
    \item $sR_i^ct$ iff $\lambda_i(s)\subseteq t$, where $\lambda_i(s)=\{\phi\mid \K_i\phi\in s\}$.
    \item $V^c(p)=\{s\in S^c\mid p\in s\}$.
\end{itemize}
\end{definition}

An equivalent definition of $R^c_i$ is as follows.
\begin{center}
    $sR^c_it$ iff for all $\phi$, if $\K_i\phi\in s$, then $\phi\in t$.
\end{center}

\begin{proposition}\label{prop.lambdasisatheory}
Where $s\in S^c$, $\lambda_i(s)$ is a consistent theory.
\end{proposition}

\begin{proof}
Suppose that $s\in S^c$. We show that $\lambda_i(s)$ is a consistent theory.
\begin{itemize}
    \item Contains $\textbf{Thm}$: given any $\phi\in \textbf{Thm}$, using $\text{GEN}$ we have $\K_i\phi\in \textbf{Thm}$. As $s\in S^c$, $\K_i\phi\in s$, hence $\phi\in\lambda_i(s)$.
    \item Closed under $\text{MP}$: suppose that $\phi,\phi\to\psi\in\lambda_i(s)$, then $\K_i\phi\in s$ and $\K_i(\phi\to\psi)\in s$. By axiom K and $s\in S^c$, we can obtain that $\K_i\psi\in s$, and therefore $\psi\in \lambda_i(s)$.
    \item Closed under $\text{RU}$: assume that $\eta(\phi\land[\psi]\neg\K_i\phi)\in \lambda_i(s)$ for all $\psi\in \EL$, to show that $\eta(U_i\phi)\in \lambda_i(s)$. By assumption, $\K_i\eta(\phi\land[\psi]\neg\K_i\phi)\in s$. Since $s$ is closed under $\text{RU}$ for the admissible form $\K_i\eta(\sharp)$, it follows that $\K_i\eta(U_i\phi)\in s$. Therefore, $\eta(U_i\phi)\in \lambda_i(s)$.
    \item Contains {\em no} $\bot$: since $s\in S^c$, $\bot\notin s$, which plus the axiom T implies that $\K_i\bot \notin s$, and thus $\bot\notin \lambda_i(s)$.
\end{itemize}
\end{proof}

\begin{proposition}\label{prop.equivalent} Let $s\in S^c$. The following conditions are equivalent.

(1) $\K_i\psi\in s$,

(2) for all $t\in S^c$, if $sR^c_it$, then $\psi\in t$.
\end{proposition}

\begin{proof}
The direction from (1) to (2) is straightforward by definition of $R^c_i$. For the converse, suppose that (1) fails (that is, $\K_i\psi\notin s$), to show that (2) fails, that is, for some $t\in S^c$, $sR^c_it$ and $\psi\notin t$.

By Prop.~\ref{prop.lambdasisatheory}, $\lambda_i(s)$ is a theory. Then by the item (1) of Prop.~\ref{prop.theory-property}, $\lambda_i(s)+\neg\psi$ is a theory and $\lambda_i(s)\cup\{\neg\psi\}\subseteq \lambda_i(s)+\neg\psi$. Moreover, by supposition, we infer that $\psi\notin \lambda_i(s)$. By Prop.~\ref{prop.lambdasisatheory}, we can obtain that $\neg\neg\psi\notin \lambda_i(s)$. Now using the item (2) of Prop.~\ref{prop.theory-property}, we derive that $\lambda_i(s)+\neg\psi$ is consistent. We have now shown that $\lambda_i(s)+\neg\psi$ is a consistent theory. Then by Lindenbaum's Lemma (Lemma~\ref{lem.lindenbaumlemma}), there exists $t\in S^c$ such that $\lambda_i(s)+\neg\psi\subseteq t$. This implies that $\lambda_i(s)\cup\{\neg\psi\}\subseteq t$. From $\lambda_i(s)\subseteq t$, it follows that $sR^c_it$; from $\neg\psi\in t$ and $t\in S^c$, it follows that $\psi\notin t$, as desired.
\end{proof}

As mentioned before, in the proof of the truth lemma below, we need to use a formula-based complexity measure to `reduce' the hard case $U_i\phi$ to the easy-to-fix case $\phi\land[\psi]\neg\K_i\phi$ (where $\psi\in\EL$). Now we introduce the notion of complexity, which combines the notions of size and $U$-depth.
\begin{definition}[Size, $U$-depth, Complexity]\label{def.complexity}

The {\em complexity} of a formula consists of two aspects: size and $U$-depth, which are defined as follows.

The {\em size} of a formula $\phi$, denoted $S(\phi)$, is a positive natural number, defined recursively as follows.
\[
\begin{array}{lll}
    S(p) & = & 1 \\
    S(\neg\phi) & = & S(\phi)+1\\
    S(\phi\land\psi) & = & S(\phi)+S(\psi)+1\\
    S(\K_i\phi) & = & S(\phi)+ 1\\
    S([\phi]\psi) & = & (5+S(\phi))\cdot S(\psi)\\
    S(U_i\phi) & = & S(\phi)+1\\
\end{array}
\]

The {\em $U$-depth} of a formula $\phi$, denoted $d_U(\phi)$, is a natural number, defined recursively as follows.
\[
\begin{array}{lll}
    d_U(p) & = & 0 \\
    d_U(\neg\phi) & = & d_U(\phi)\\
    d_U(\phi\land\psi) & = & max\{d_U(\phi),d_U(\psi)\}\\
    d_U(\K_i\phi) & = & d_U(\phi)\\
    d_U([\phi]\psi) & = & d_U(\phi)+d_U(\psi)\\
    d_U(U_i\phi) & = & d_U(\phi)+1\\
\end{array}
\]

We say that $\phi$ is {\em less complex than} $\psi$, denoted $c(\phi)<c(\psi)$, if either $d_U(\phi)<d_U(\psi)$, or $d_U(\phi)=d_U(\psi)$ and $S(\phi)<S(\psi)$.
\end{definition}

One may easily check that $S(\phi\to\psi)=S(\phi)+S(\psi)+3$ and $d_U(\phi\to\psi)=max\{d_U(\phi),d_U(\psi)\}$. Moreover, $d_U(\phi)=0$ for all $\phi\in\EL$.

Note that our definition of $S(\phi\land\psi)$ is different from that in~\cite{balbianietal:2015,DitmarschF16,Fan:2016}. The reason of so defining is that, for us, it is more suitable for the intuition of the size of a conjunction. Also note that the number 5 in the definition of $S([\phi]\psi)$ is the least natural number that gives us the following properties (more precisely, (4) and (5) in Prop.~\ref{prop.properties}), in contrast to the least natural number 4 in the complexity measure for public announcement logic in~\cite[Chap.~7]{hvdetal.del:2007}.

\begin{proposition}\label{prop.properties}\
\begin{enumerate}
    \item[(1)] $c(\psi)<c(\neg\psi)$
    \item[(2)] $c(\psi)<c(\psi\land\chi)$, $c(\chi)<c(\psi\land\chi)$.
    \item[(3)] $c(\psi)<c(\K_i\psi)$
    \item[(4)] $c(\psi\to p)< c([\psi]p)$
    \item[(5)] $c(\psi\to\neg[\psi]\chi)<c([\psi]\neg\chi)$
    \item[(6)] $c([\psi]\chi_1)<c([\psi](\chi_1\land\chi_2))$, $c([\psi]\chi_2)<c([\psi](\chi_1\land\chi_2))$
    \item[(7)] $c(\psi\to{\K_i[\psi]\chi})<c({[\psi]\K_i}\chi)$
    \item[(8)] $c([\psi\land[\psi]\chi]\delta)<c([\psi][\chi]\delta)$
    \item[(9)] $c([\psi]\chi)<c([\psi]U_i\chi)$,$c([\psi]{[\delta]\neg\K_i}\chi)<c([\psi]U_i\chi)$, where $\delta\in\EL$.
    \item[(10)] $c(\psi)<c(U_i\psi)$, $c({[\chi]\neg\K_i}\psi)<c(U_i\psi)$, where $\chi\in\EL$.
\end{enumerate}
\end{proposition}

\begin{proof}
We show some of them as examples.

(4) One may easily verify that $d_U(\psi\to p)=d_U(\psi)=d_U([\psi]p)$. Moreover, $S(\psi\to p)=S(\psi)+S(p)+3=S(\psi)+4<5+S(\psi)=S([\psi]p)$. Thus $c(\psi\to p)<c([\psi]p)$.

(5) One may easily verify that $d_U(\psi\to\neg[\psi]\chi)=d_U(\psi)+d_U(\chi)=d_U([\psi]\neg\chi)$. Moreover, $S(\psi\to\neg[\psi]\chi)=S(\psi)+5\cdot S(\chi)+S(\psi)\cdot S(\chi)+4<S(\psi)+5\cdot S(\chi)+S(\psi)\cdot S(\chi)+5=S([\psi]\neg\chi)$. Therefore, $c(\psi\to\neg[\psi]\chi)<c([\psi]\neg\chi)$.

(7) One may easily verify that $d_U(\psi\to{\K_i[\psi]}\chi)=d_U(\psi)+d_U(\chi)=d_U({[\psi]\K_i}\chi)$. Moreover, $S(\psi\to{\K_i[\psi]}\chi)=S(\psi)+5\cdot S(\chi)+S(\psi)\cdot S(\chi)+4<S(\psi)+5\cdot S(\chi)+S(\psi)\cdot S(\chi)+5=S({[\psi]\K_i}\chi)$. Therefore, $c(\psi\to{\K_i[\psi]}\chi)<c({[\psi]\K_i}\chi)$.

(8) One may easily verify that $d_U([\psi\land[\psi]\chi]\delta)=d_U(\psi)+d_U(\chi)+d_U(\delta)=d_U([\psi][\chi]\delta)$. Moreover, $S([\psi\land[\psi]\chi]\delta)=6\cdot S(\delta)+S(\psi)\cdot S(\delta)+ 5\cdot S(\chi)\cdot S(\delta)+S(\psi)\cdot S(\chi)\cdot S(\delta)$, and $S([\psi][\chi]\delta)=25\cdot S(\delta)+5\cdot S(\psi)\cdot S(\delta)+5\cdot S(\chi)\cdot S(\delta)+S(\psi)\cdot S(\chi)\cdot S(\delta)$, thus $S([\psi\land[\psi]\chi]\delta)<S([\psi][\chi]\delta)$. Therefore, $c([\psi\land[\psi]\chi]\delta)<c([\psi][\chi]\delta)$.


(10) Since $d_U(\psi)< d_U(\psi)+1= d_U(U_i\psi)$, we obtain $c(\psi)<c(U_i\psi)$. Also, one may verify that $d_U({[\chi]\neg\K_i}\psi)=d_U(\psi)<d_U(\psi)+1=d_U(U_i\psi)$, thus $c({[\chi]\neg\K_i}\psi)<c(U_i\psi)$.
\end{proof}

Now we are ready to show the truth lemma.
\begin{lemma}
For all $s\in S^c$, for all $\phi\in \textbf{LUT}$, we have
$$\M^c,s\vDash\phi\text{ iff }\phi\in s.$$
\end{lemma}

\begin{proof}
By induction on the complexity of $\phi$. Recall that the notion of complexity is given in Def.~\ref{def.complexity}.
\begin{itemize}
    \item $\phi=p\in\BP$. Straightforward by the definition of $V^c$.
    \item $\phi=\neg\psi$. By the item (1) of Prop.~\ref{prop.properties}, $c(\psi)<c(\phi)$. Then by induction hypothesis, we have that $\M^c,s\vDash\psi$ iff $\psi\in s$. Thus $\M^c,s\vDash\phi$ iff $\M^c,s\nvDash\psi$ iff $\psi\notin s$ iff $\phi\in s$.
    \item $\phi=\psi\land\chi$. By the item (2) of Prop.~\ref{prop.properties}, $c(\psi),c(\chi)<c(\phi)$. Then by induction hypothesis, $\M^c,s\vDash\psi$ iff $\psi\in s$, and $\M^c,s\vDash \chi$ iff $\chi\in s$. Therefore, $\M^c,s\vDash \phi$ iff ($\M^c,s\vDash\psi$ and $\M^c,s\vDash\chi$) iff ($\psi\in s$ and $\chi\in s$) iff $\phi\in s$.
    \item $\phi=\K_i\psi$. By the item (3) of Prop.~\ref{prop.properties}, $c(\psi)<c(\phi)$. Then we have the following equivalences.
    \[
    \begin{array}{lll}
        \M^c,s\vDash\K_i\psi & \iff &  \text{for all }t\in S^c,\text{ if }sR^c_it, \text{ then }\M^c,t\vDash\psi \\
         & \iff & \text{for all }t\in S^c,\text{ if }sR^c_it, \text{ then }\psi\in t \\
         & \iff & \K_i\psi\in s,\\
    \end{array}
    \]
where the penultimate equivalence holds by induction hypothesis, whereas the last one has been shown in Prop.~\ref{prop.equivalent}.
    \item $\phi=[\psi]p$. By the item (4) of Prop.~\ref{prop.properties}, $c(\psi\to p)<c(\phi)$. Then we have the following equivalences.
    \[
    \begin{array}{lll}
        \M^c,s\vDash [\psi]p &\iff & \M^c,s\vDash \psi\to p \\
         & \iff & \psi\to p\in s\\
         & \iff & [\psi]p\in s,\\
    \end{array}
    \]
where the penultimate equivalence holds by induction hypothesis, whereas other two equivalences holds by axiom $\text{AP}$.
    \item $\phi=[\psi]\neg\chi$. By the item (5) of Prop.~\ref{prop.properties}, $c(\psi\to\neg[\psi]\chi)<c(\phi)$. Then we have
    \[
    \begin{array}{lll}
        \M^c,s\vDash[\psi]\neg\chi & \iff & \M^c,s\vDash\psi\to \neg[\psi]\chi \\
         & \iff & \psi\to\neg[\psi]\chi\in s \\
         & \iff & [\psi]\neg\chi\in s,\\
    \end{array}
    \]
where the penultimate equivalence holds by induction hypothesis, whereas other two equivalences holds by axiom $\text{AN}$.
    \item $\phi=[\psi](\chi_1\land\chi_2)$. By the item (6) of Prop.~\ref{prop.properties}, $c([\psi]\chi_1),c([\psi]\chi_2)<c(\phi)$. Then we have
    \[
    \begin{array}{lll}
        \M^c,s\vDash [\psi](\chi_1\land\chi_2) & \iff & \M^c,s\vDash [\psi]\chi_1\text{ and }\M^c,s\vDash[\psi]\chi_2 \\
         & \iff & [\psi]\chi_1\in s\text{ and }[\psi]\chi_2\in s \\
         & \iff & [\psi](\chi_1\land\chi_2)\in s,\\
    \end{array}
    \]
where the penultimate equivalence holds by induction hypothesis, whereas other two equivalences holds by axiom $\text{AC}$.
    \item $\phi={[\psi]\K_i}\chi$. By the item (7) of Prop.~\ref{prop.properties}, $c(\psi\to{\K_i[\psi]}\chi)<c(\phi)$. Then we have
    \[
    \begin{array}{lll}
        \M^c,s\vDash{[\psi]\K_i}\chi & \iff & \M^c,s\vDash (\psi\to{\K_i[\psi]}\chi) \\
        & \iff & \psi\to{\K_i[\psi]}\chi\in s\\
        & \iff & {[\psi]\K_i}\chi\in s,\\
    \end{array}
    \]
where the penultimate equivalence holds by induction hypothesis, whereas other two equivalences holds by axiom $\text{AK}$.
    \item $\phi=[\psi][\chi]\delta$. By the item (8) of Prop.~\ref{prop.properties}, $c([\psi\land[\psi]\chi]\delta)<c(\phi)$. Then we have
    \[
    \begin{array}{lll}
        \M^c,s\vDash[\psi][\chi]\delta & \iff & \M^c,s\vDash[\psi\land[\psi]\chi]\delta \\
        & \iff & [\psi\land[\psi]\chi]\delta\in s\\
         & \iff & [\psi][\chi]\delta\in s, \\
    \end{array}
    \]
where the penultimate equivalence holds by induction hypothesis, whereas other two equivalences holds by axiom $\text{AA}$. 
    \item $\phi=[\psi]U_i\chi$. By the item (9) of Prop.~\ref{prop.properties}, $c([\psi]\chi)<c([\psi]U_i\chi)$,$c([\psi]{[\delta]\neg\K_i}\chi)<c([\psi]U_i\chi)$, where $\delta\in\EL$. Then we have
    \[
    \begin{array}{lll}
        \M^c,s\vDash[\psi]U_i\chi & \iff & \M^c,s\vDash\psi\text{ implies }\M^c|_\psi,s\vDash U_i\chi \\
        & \iff & \M^c,s\vDash\psi\text{ implies }(\M^c|_\psi,s\vDash\chi\text{ and for all }\delta\in\EL,\M^c|_\psi,s\vDash{[\delta]\neg\K_i}\chi)\\
        & \iff & (\M,s\vDash\psi\text{ implies }\M^c|_\psi,s\vDash\chi) \text{ and }\\
        & & (\M,s\vDash\psi\text{ implies for all }\delta\in\EL,\M^c|_\psi,s\vDash{[\delta]\neg\K_i}\chi)\\
        & \iff & \M,s\vDash[\psi]\chi\text{ and for all }\delta\in\EL,\M^c,s\vDash[\psi]{[\delta]\neg\K_i}\chi\\
        &\stackrel{(\ast)}{\iff} & [\psi]\chi\in s\text{ and for all }\delta\in\EL, [\psi]{[\delta]\neg\K_i}\chi\in s\\
        & \stackrel{(\ast\ast)}{\iff} & [\psi](\chi\land{[\delta]\neg\K_i}\chi)\in s\text{ for all }\delta\in\EL\\
        & \stackrel{(\ast\ast\ast)}{\iff} & [\psi]U_i\chi\in s,\\
    \end{array}
    \]
where $(\ast)$ holds by induction hypothesis, $(\ast\ast)$ holds due to axiom $\text{AC}$ and $s\in S^c$, and $(\ast\ast\ast)$ uses the fact that $s$ is closed under $\text{RU}$ for the admissible form $[\psi]\sharp$, axiom $\text{AU}$, and the rule $\text{RM}[\cdot]$, which is in turn derivable from the axiom KA and the rule GENA.
    \item $\phi=U_i\psi$. By the item (10) of Prop.~\ref{prop.properties}, $c(\psi)<c(U_i\psi)$, $c({[\chi]\neg\K_i}\psi)<c(U_i\psi)$, where $\chi\in\EL$. Then we have
    \[
    \begin{array}{lll}
        \M^c,s\vDash U_i\psi & \iff & \M^c,s\vDash\psi\text{ and for all }\chi\in\EL,\M^c,s\vDash{[\chi]\neg\K_i}\psi \\
         & \iff & \psi\in s\text{ and }{[\chi]\neg\K_i}\psi\in s\text{ for all }\chi\in\EL \\
         & \iff & \psi\land{[\chi]\neg\K_i}\psi\in s\text{ for all }\chi\in\EL\\
         & \iff & U_i\psi\in s,\\
    \end{array}
    \]
where the second equivalence holds by induction hypothesis, the third one uses the closure of $s$ under $\text{MP}$, and the last one uses the closure of $s$ under $\text{RU}$ for the admissible form $\sharp$ and the axiom $\text{AU}$.
\end{itemize}
\end{proof}

It is now a standard exercise to show the following result.
\begin{theorem}\label{thm.comp-slut}
$\mathbb{LUT}$ is sound and complete with respect to the class of all frames. That is, if $\vDash\phi$, then $\vdash\phi$.
\end{theorem}

With the completeness in hand, we can give a syntactic proof for Fitch's paradox of knowability, which is much simpler than those in the literature, e.g. in~\cite{Routley:2010}.
\begin{proposition}
If $\vdash\neg U_i\phi$ for all $\phi$, then $\vdash \neg{\bullet_i\phi}$ for all $\phi$.
\end{proposition}

\begin{proof}
    Suppose that $\vdash\neg U_i\phi$ for all $\phi$, then $\vdash \neg U_i{\bullet_i\phi}$. Then by Prop.~\ref{prop.unknowntruthsareunknowable} and Thm.~\ref{thm.comp-slut}, $\vdash \bullet_i\phi\to U_i{\bullet_i\phi}$, and therefore $\vdash \neg{\bullet_i\phi}$ for all $\phi$.
\end{proof}

\section{Conclusion and Discussion}\label{sec.conclusion}

In this paper, we proposed a logic of unknowable truths, investigated the logical properties of unknown truths and the logical properties of unknowable truths, which includes the similarities and relationship between the two notions, and finally axiomatized the logic of unknowable truths.

So far we have seen that the semantics of unknowable truths depends on that of propositional knowledge. However, we have mainly focused on the relationship between unknowable truths and unknown truths, instead of between unknowable truths and propositional knowledge, and the notion of unknown truths is weaker than that of propositional knowledge. So there seems to be some asymmetry between the semantics of unknowable truths and our concerns. 
Then a natural question is: is there any other semantics for unknowable truths which relates the notion to that of unknown truths more directly/properly? Our answer is positive. We will introduce this new semantics in the future work.

For another future work, a natural extension would be to add propositional quantifiers. This can increase the expressive power of the current logic. For instance, we can express Fitch's paradox of knowability in the new language as follows: $\forall p\neg U_ip\to \forall p\neg {\bullet_ip}$, or equivalently, $\exists p{\bullet_ip}\to \exists pU_ip$.

\bibliographystyle{plain}
\bibliography{biblio2023}

\end{document}